\documentclass[copyright,creativecommons]{eptcs}
\providecommand{\event}{EXPRESS/SOS 2012} % Name of the event you are submitting to
\usepackage{breakurl}             % Not needed if you use pdflatex only.

\usepackage[final]{fixme}

\usepackage[english]{babel}
\usepackage{amsmath,amssymb,amsthm}
\usepackage{verbatim}
\usepackage{txfonts}
\usepackage{times} % condensed font for more pages
\usepackage{enumerate}
\usepackage[all]{xy}
\SelectTips{cm}{}

\let\amalg=\undefined
\let\coprod=\undefined
\let\jmath=\undefined
\DeclareSymbolFont{cmsymbols}{OMS}{cmsy}{m}{n}
\DeclareSymbolFont{cmlargesymbols}{OMX}{cmex}{m}{n}
\DeclareMathSymbol{\amalg}{\mathbin}{cmsymbols}{"71}
\DeclareMathSymbol{\coprod}{\mathop}{cmlargesymbols}{"60}
\DeclareSymbolFont{cmletters}{OML}{cmm}{m}{it}
\DeclareMathSymbol{\jmath}{\mathord}{cmletters}{"7C}

\newcommand{\takeout}[1]{\empty}

\newtheorem{theorem}{Theorem}[section]
\newtheorem{proposition}[theorem]{Proposition}
\newtheorem{corollary}[theorem]{Corollary}
\newtheorem{lemma}[theorem]{Lemma}
\theoremstyle{definition}
\newtheorem{definition}[theorem]{Definition}
\newtheorem{assumption}[theorem]{Assumption}
\newtheorem{example}[theorem]{Example}
\newtheorem{remark}[theorem]{Remark}

\numberwithin{equation}{section}

\def\refeq#1{(\ref{#1})}

\newcommand*{\longhookrightarrow}{\ensuremath{\lhook\joinrel\relbar\joinrel\rightarrow}}

\let\oldrho=\rho
\renewcommand{\rho}{\varrho}

\def\set{\mathsf{Set}}
\def\jsl{\mathsf{Jsl}}
\def\coalg{\mathsf{Coalg}}
\def\coalgf{\coalg_{\mathsf{f}}}
\def\A{\mathcal{A}}
\def\Vec{\mathsf{Vec}}
\def\VecR{\Vec_\real}

\def\powf{\mathcal{P}_\mathsf{f}}
\def\pow{\mathcal{P}}

\def\real{\mathbb{R}}
\def\nat{\mathbb{N}}
\def\F{\mathbb{F}}
\def\M{\mathbb{M}}

\def\o{\cdot}
\def\inr{\mathsf{inr}}
\def\inl{\mathsf{inl}}
\def\eps{\varepsilon}
\def\id{\mathit{Id}}

\title{On the specification of operations\\on the rational behaviour of systems}

\author{Marcello M.~Bonsangue
\institute{LIACS, Leiden University}
\institute{The Netherlands}
\email{marcello@liacs.nl}
\and
 Stefan Milius
\institute{Technische Universit\"at Braunschweig\\ Germany}
\email{mail@stefan-milius.eu}
\and
Jurriaan Rot\thanks{This author is supported by the NWO project CoRE.}
\institute{LIACS, Leiden University}
\institute{The Netherlands}
\email{jrot@liacs.nl}
}
\def\titlerunning{Operations on rational behaviour}
\def\authorrunning{M.~M.~Bonsangue, S.~Milius, J.~Rot}

\begin{document}
%
% For FIXME comments
%
\FXRegisterAuthor{mb}{amb}{Mar}%Marcello
\FXRegisterAuthor{sm}{asm}{Ste}%Stefan
\FXRegisterAuthor{jr}{ajr}{Jur}%Jurriaan

\maketitle

\begin{abstract}
Structural operational semantics can be studied at the general level of
%coalgebras in terms of 
distributive laws of syntax over behaviour. This yields specification
formats for well-behaved algebraic operations on final coalgebras,
which are a domain for the behaviour of all systems of a given type
functor. We introduce a format for specification of algebraic
operations that restrict to the \emph{rational fixpoint} of a functor,
which captures the behaviour of \emph{finite} systems. %of the given type functor. 
In other words, we show that rational behaviour is closed
under operations specified in our format. As applications we consider
operations on regular languages, regular processes and finite weighted
transition systems.
\end{abstract}

\section{Introduction}
\emph{Structural operational semantics} (SOS) is a popular and widely used framework for
defining operational semantics by means of transition system specifications. Syntactic
restrictions on the format of these specifications give rise to algebraic properties
of operations on system behaviour~\cite{AFV}, e.\,g.,~GSOS
rules~\cite{BIM95} ensure that bisimilarity is a congruence. 

The key insight to give a uniform mathematical treatment of various
flavours of SOS is that the theory of \emph{coalgebras} provides a common
framework for the study of state-based systems and their
behaviour. This includes labelled transition systems but also stream automata, 
(non-)deterministic automata, weighted transition systems and many
more. The type of a coalgebra is expressed by an {endo}\-functor $F$, and
a canonical domain for system behaviour is provided by the final
$F$-coalgebra.

Turi and Plotkin~\cite{TP97} show in their seminal paper that the
interplay between syntax and behaviour given by transition system
specifications can be generalized by \emph{distributive laws} of a
functor $\Sigma$, representing the syntax, over a functor $F$,
representing the behaviour. They formulate and prove that bisimilarity
is a congruence at this level of generality. The final $F$-coalgebra here plays an important r\^{o}le as the
denotational model of a transition system specification. In
particular, a distributive law induces a canonical $\Sigma$-algebra
structure on the final coalgebra for $F$.

\takeout{ % old paragraph
The theory of \emph{coalgebras} provides a general categorical framework for the study of 
state-based systems, including labelled transition systems but also stream automata, 
(non-)deterministic automata and weighted transition systems. 
The type of a coalgebra is expressed in terms of an endofunctor $F$. This allows
for a uniform study of system behaviour. In fact, under some mild conditions 
on the functor $F$, a canonical domain of behaviour exists and is given by
the so-called \emph{final coalgebra}.

The interplay between syntax and behaviour given by formats on transition system specifications
is generalized to the theory of coalgebra by \emph{distributive laws} of a functor $\Sigma$, 
representing the syntax, over a functor $F$, representing the behaviour. More precisely, 
such a distributive law defines a $\Sigma$-algebra structure on the final coalgebra of 
$F$~\cite{TP97,Klin11}. 
} % end takeout

But the final $F$-coalgebra is the domain of the behaviour of
\emph{all} $F$-coalgebras, and often it is interesting to study the
behaviour of only \emph{finite-state} systems, such as finite automata
or regular processes. In fact, finite-state systems have nice
decidability properties and are amenable to automated verification
techniques.  The \emph{rational fixpoint} of a set functor $F$ is the
subcoalgebra of the final coalgebra given by the behaviours of all
finite coalgebras~\cite{AMV06,Milius10}.  For example, regular
languages, rational streams~\cite{Rutten08}, rational formal power
series~\cite{DKV09} %\smnote{Other reference needed?}
and regular trees for a signature~\cite{Courcelle83} form rational
fixpoints of appropriate functors $F$.

In this paper we investigate \emph{bipointed specifications}, a
restricted type of distributive laws which induces operations on the
rational fixpoint of a functor $F$ as a restriction of the same
operations on the final coalgebra.  As a result we show that regular
system behaviour is closed under operations induced by bipointed
specifications. So this yields an easy syntactic criterion to check
that regular behaviour is closed under certain algebraic
operations. Applications include operations on regular languages and
finite automata, such as the well-known shuffle operator, operations
on finite weighted transition systems and regular processes.

%on regular system behaviour, i.e., we show a rule format that defines operations on the final coalgebra $\nu F$
%that restricts to the rational fixpoint $\rho
%F$~\cite{AMV06,Milius10}. The rational fixpoint of a set functor $F$ is the
%subcoalgebra of $\nu F$ given by all the behaviours of finite
%coalgebras.

There is a large body of work on SOS formats and distributive laws
(see~\cite{Klin11} for a good overview). Bipointed specifications appear (without a
name) as an intermediate format between abstract toy SOS~\cite{Klin07}
and the abstract operational rules of~\cite{TP97}. However, we are not
aware of any work on formats for finite coalgebras. 
The only exception is the work on labelled transition systems by
Aceto~\cite{aceto94} (see also~\cite{AFV}). When instantiated on coalgebras corresponding to
labelled transition systems, bipointed specifications coincide with
specifications in the \emph{simple GSOS} format of loc.~cit.~on finite signatures. Our contribution
can thus be seen as a generalization of the simple GSOS format to
the realm of distributive laws. In~\cite{aceto94,AFV} there is also an
extension to countable signatures with certain finite dependencies
among the operators, and it is proved that the labelled transition
system induced by a simple GSOS specification is regular, i.\,e., for
each closed process term $P$ the ensuing transition system defining the
operational semantics of $P$ has finitely many states
(see~\cite[Theorem~5.28]{AFV}). In future work we shall incorporate
such a result in our theory. 
  
  The outline of this paper is as follows. In the next section we introduce
  the necessary preliminaries. Then in Section~\ref{sec-sos} we present
  our specification format. This induces an algebra on the rational
  fixpoint, as shown in Section~\ref{sec-alg}. We proceed in 
  Section~\ref{sec-app} with several applications of the theory, 
  and we finish in Section~\ref{sec-conc} with conclusions and
  suggestions for future work. 

%\jrnote{Stefan: you wrote something about Power, Lenisa and Watanabe,
%would you like to include that?}
%\smnote{No, I think it is not necessary.}

\section{Preliminaries}\label{sec-prel}

We assume that the reader is familiar with basic notions of category
theory. With $\set$ we denote the category of sets and functions. In
any category we write products and coproducts with their projections
and injections, respectively as
$
\xymatrix@1{
  A & A\times B \ar[l]_-{\pi_0} \ar[r]^-{\pi_1} & B
}
%\qquad
\ \text{and}\ 
%\qquad
\xymatrix@1{
   C \ar[r]^-\inl & C + D & D \ar[l]_-\inr}.
$
The corresponding unique induced morphisms are denoted $\langle a,b\rangle: E
\to A \times B$ and $[c,d]: C + D \to E$.

\subsection{Algebras and coalgebras}

Let $\mathcal{A}$ be a category and $F: \mathcal{A} \rightarrow \mathcal{A}$ a functor. 
An \emph{$F$-algebra} is a pair $(A, \alpha)$ where $A$ is an object of $\mathcal{A}$ called
the \emph{carrier} and $\alpha: FA \rightarrow A$ is a morphism called
the \emph{structure} of the algebra.  Given algebras $(A, \alpha)$ and $(B, \beta)$, an
\emph{algebra homomorphism} is a map $f: A \rightarrow B$ such that 
$f \circ \alpha = Ff \circ \beta$.
A \emph{signature} is a set $\Sigma$ of operation symbols with
prescribed arity $|\sigma| \in \nat$ for each $\sigma \in \Sigma$.
This can equivalently be represented as a \emph{polynomial} functor 
\[
\Sigma X = \coprod_{\sigma \in \Sigma} X^{|\sigma|}.
\]
(We shall abuse notation and denote by $\Sigma$ both a signature and
its corresponding polynomial functor.)
For example, a signature $\Sigma_0$ on $\set$ consisting of a binary
operation symbol $b$ and a constant symbol $c$ corresponds to the functor $\Sigma_0 X
= X \times X + 1$. A $\Sigma_0$-algebra then is a set $A$ together
with an actual binary operation $b_A : A \times A \to A$ and a
constant $c_A \in A$, and algebra homomorphisms are precisely the maps
between algebras preserving the binary operation and the constant. 
\begin{example} 
 A \emph{join-semilattice} is a set $S$ with a binary operator
$\vee: S \times S \rightarrow S$ called the \emph{join}, and 
an element $\bot \in S$ (or $\bot: 1 \rightarrow S$) called \emph{bottom}; 
equivalently, it is an algebra $[\vee, 0]: S \times S + 1 \rightarrow S$. 
The join is associative, commutative and idempotent, and the bottom
is the  identity element with respect to the join. With 
$\jsl$ we denote the category of join-semilattices and homomorphisms
between them.
\end{example}

An $F$-\emph{coalgebra} is a pair $(S, f)$ such that $S$
is an object of $\mathcal{A}$, called the \emph{carrier}, and $f: S \rightarrow FS$
is an arrow, called the \emph{transition structure} or \emph{dynamics}. For 
coalgebras $(S, f_S)$ and $(T, f_T)$, a 
\emph{coalgebra homomorphism} is a morphism $h: S \rightarrow T$ such that 
$f_T \circ h = Fh \circ f_S$. If $\mathcal{A} = \set$ and $(S,f_S)$
and $(T,f_T)$ are coalgebras, then a \emph{bisimulation} is a relation
$R \subseteq S \times T$ such that $R$ carries a coalgebra structure $f_R$
and the projection maps $\pi_0: R \rightarrow S$
and $\pi_1: R \rightarrow T$ are coalgebra homomorphisms from
$(R,f_R)$ to $(S, f_S)$ and $(T, f_T)$, respectively. 
%For other categories, a similar definition of bisimulation is obtained using a generalisation the notion of a relation.
We denote by 
\[\coalg(F)\] 
the category of $F$-coalgebras and their homomorphisms. Of special interest are
\emph{final coalgebras}, i.\,e., final objects of categories
$\coalg(F)$, which exist under mild conditions on $F$. 
Thus, if a category $\coalg(F)$ has a final coalgebra $(\nu F, t)$, then there
exists, for each $F$-coalgebra $(S, f)$ a unique coalgebra
homomorphism $f^{\dagger}: S \rightarrow \nu F$. A final coalgebra is
determined uniquely up to isomorphism. Moreover, by the
famous Lambek Lemma~\cite{lambek}, the transition structure $t: \nu F \to
F(\nu F)$ is an isomorphism. 
The final coalgebra can be thought of as a canonical domain of behaviour of the type of
systems corresponding to the functor $F$. We consider several
examples.
\takeout{ % This was a trial - but need weak preservation of pullbacks
          % to say "modulo largest bisimulation. 
\begin{remark}
  Most functors of interest in this paper are \emph{finitary} functors
  of $\set$, i.\,e., they are determined by their behaviour on finite
  sets. More precisely, a functor $F: \set \rightarrow \set$ is
  finitary if it preserves filtered colimits; equivalently, $F$ is
  \emph{bounded} (see, e.\,g.,~Ad\'amek and Trnkov\'a~\cite{at}),
  i.\,e., for every set $X$ and every element $x \in FX$, there is a
  finite subset $i: Y\hookrightarrow X$ such that
  $x \in Fi[FY] \subseteq FX$. 

  For a finitary functor $F$ on $\Set$ the final coalgebras can be
  described is
  essentially given as the coproduct of all $F$-coalgebras modulo the
  greatest bisimulation ....
\end{remark}
}

\begin{example}\label{ex:coalg}
  \begin{enumerate}[(1)]
  \item Coalgebras for the functor $FX = \real \times X$ on $\set$,
    where $\real$ is the set of real numbers, are often called
    \emph{stream systems} over the reals. The carrier of the final $F$-coalgebra is
    the set $\real^\omega =\{\sigma \mid \sigma: \nat \rightarrow
    \real\}$ of all streams (infinite sequences) of elements of
    $\real$. The transition structure $\langle o, t \rangle:
    \real^\omega \rightarrow \real \times \real^\omega$ is defined as
    $o(\sigma) = \sigma(0)$ and $t(\sigma)(n) = \sigma(n+1)$.
  \item Deterministic automata with input alphabet $A$ are coalgebras
    for the functor $FX = 2 \times X^A$, where $2 = \{0,1\}$. Indeed,
    to give a coalgebra $f: S \to 2 \times S^A$ precisely corresponds
    to giving a set $S$ of states with a map 
    $o: S \to 2$ (indicating final states) and a map $t: S
    \to S^A$, where $t(s)(a)$ is the successor of state $s$ under input
    $a$. The final coalgebra is carried by the set of all formal
    languages $\pow(A^*)$ with its coalgebra structure given by $o:
    \pow(A^*) \to 2$ with $o(L) = 1$ iff $L$ contains the empty word
    and $t: \pow(A^*) \to \pow(A^*)^A$ given by the language derivative
    $t(L)(a) = \{\,w \mid aw \in L\,\}.$
    For a given automaton $(S, f)$ the unique coalgebra homomorphism
    maps a state to the language it accepts. 
    
  \item Labelled transition systems (LTS) with actions from the set $A$ are coalgebras for the functor $FX
    = \powf(A \times X)$. Indeed, a coalgebra $f: X \to \powf(A \times
    X)$ corresponds precisely to giving a set $X$ of states and a
    transition relation $R \subseteq X \times A \times X$ that is
    \emph{finitely branching}, i.\,e., for every $x \in X$ there are only finitely
    many $a \in A$ and $x' \in S$ with $(x,a,x') \in R$. The final
    coalgebra for $F$ exists and can be thought of as consisting of
    processes modulo strong bisimilarity of Milner~\cite{milner}. More
    precisely, it follows from~\cite[Proposition~5.16]{AMV06}
    (cf.~also Barr~\cite{barr_coalg}) that the final coalgebra is the coproduct
    of all countable $F$-coalgebras modulo the greatest
    bisimulation.\footnote{This can be thought of as the coproduct of
      \emph{all} coalgebras modulo the greatest bisimulation; but this
      coproduct is a proper class, whence the restriction to countable
    coalgebras.}

  \item A very similar example are non-deterministic automata with a
    finite input alphabet $A$. They are coalgebras for $FX = 2 \times
    (\powf X)^A$. Here the final coalgebra consists of all behaviours
    modulo bisimilarity of non-deterministic automata; more precisely,
    $\nu F$ is the coproduct of all countable $F$-coalgebras modulo
    the largest bisimulation as in the previous point.  
    A (necessarily) isomorphic description of $\nu F$ follows from the description of
    the final coalgebra for $\powf$ given by Worrell~\cite{worrell};
    (see also~\cite{BMS12}): the elements of $\nu F$ are
    finitely branching \emph{strongly extensional} trees with edges
    labelled in $A$ and nodes labelled in $2$. Due to lack of space we
    omit recalling the definition of a strongly extensional tree and
    refer the reader to~\cite{worrell,BMS12} instead. 
    %But note that a similar description is possible in the previous point~(3).

  \item Weighted transition systems (WTS) are labelled transition
    systems where transitions have weights (modelling multiplicities,
    costs, probabilities, etc.). We consider WTS's where the weights
    are elements of a commutative monoid $\M = \langle M, +, 0
    \rangle$. In order to define them coalgebraically as done
    in~\cite{Klin09}, we first consider the $\set$ endofunctor
    $\mathcal{F}_{\M}$, which acts on a set $X$ and a function $f: X
    \rightarrow Y$ as
    $$\mathcal{F}_{\M}(X) = \{\phi: X \rightarrow M \mid \phi \text{ has finite support}\}
      \qquad 
      \mathcal{F}_{\M} f(\phi)(y) = \sum_{x \in f^{-1}(y)} \phi(x),
    $$
    where a function $\phi: X \to M$ has finite support if $\phi(x) \neq 0$ for finitely many $x \in X$. A \emph{weighted transition system}
    is a coalgebra for the functor $FX = (\mathcal{F}_{\M} X)^A$ for a set of labels $A$. The final $F$-coalgebra
    exists for any monoid $\M$. Similarly as before, it is the coproduct of all countable $F$-coalgebras modulo weighted bisimilarity of~\cite{Klin09}.

    %\smnote{I took out ``obs. equiv.''}
    \takeout{Two elements $x, y$ of the carrier of an $F$-coalgebra are called \emph{observationally equivalent}
    if they are mapped to the same element in the final coalgebra. Observational equivalence coincides with 
    so-called \emph{weighted bisimilarity}~\cite{Klin09}.
  }

\takeout{
  \item Weighted automata were invented by
    Sch\"utzenberger~\cite{schuetzenberger} (see also~\cite{DKV09}). They are non-deterministic automata
    where transitions are labelled by weights (modelling multiplicities, costs,
    probabilities, etc.). Usually one takes the weights as elements
    from a semiring---here we consider only weights from a field
    $\F$ for simplicity. A weighted automaton with input alphabet $A$ is then given by a
    family of $n \times n$ transition matrices $M_a$, $a \in A$, and by an output vector $o$ of
    length $n$ over $\F$. This represents $n$ states, where each state
    $i$ has output $o(i)$ and $M_a(i,j)$ is the weight of the
    transition from $i$ to $j$ under input $a$. Equivalently, we have
    linear maps $o: \F^n \to \F$ and $M_a: \F^n \to \F^n$, $a \in
    A$. Combining these we see that a weighted automaton is a
    coalgebra $\F^n \to \F \times (\F^n)^A$ for the functor $FX = \F
    \times X^A$ on the category $\Vec_\F$ of vector spaces over $\F$. 

    It is well-known that every state of a weighted automaton accepts
    a \emph{formal power series} (or weighted language) viz.~an element of
    $\F^{A^*}$. Indeed, for all $s \in \F^n$ we obtain $L_s: A^* \to
    \F$ by simultaneous induction on the length of $w \in A^*$:
    \[
    L_s(\eps) = o(s)\qquad L_s(aw) = L_{M_a(s)}(w)
    \qquad\text{for all $s \in \F^n$.}
    \]

    Notice that $\F^{A^*}$ carries the canonical componentwise
    structure of a vector space and a natural coalgebra structure given
    by $o: \F^{A^*} \to \F$ with $o(L) = L(\eps)$ and $t: \F^{A^*} \to
    (\F^{A^*})^A$ with $t(L)(a) = \lambda w. L(aw)$. This turns
    $\F^{A^*}$ into a final $F$-coalgebra and for every weighted
    automaton the unique homomorphism into $\nu F$ maps a state to its
    accepted formal power series as defined above. 

    Notice that the classical weighted automata are precisely those
    $F$-coalgebras with a finite dimensional carrier $\F^n$. So one
    expects that not every formal power series can be accepted by a
    weighted automaton; in fact, the (finite dimensional) weighted
    automata accept precisely the \emph{rational} formal power series
    (see~\cite{DKV09}).
}

  \item Let $F = \Sigma$ be a polynomial functor on $\set$. The final
    coalgebra $\nu F$ is carried by the set of all (finite and
    infinite) $\Sigma$-trees, i.\,e., rooted and
    ordered trees labelled in the signature $\Sigma$ so that inner nodes with
    $n$ children are labelled by $n$-ary operation symbols and
    leaves are labelled by constant symbols. The coalgebra structure
    of $\nu F$ is given by the inverse of tree-tupling. 
  \end{enumerate}
  % (to do: deterministic/non-deterministic automata, weighted
%automata)\smnote{I'd do LTS's instead of non-det. autom.}
\end{example}

\subsection{Locally finitely presentable coalgebras} 
We are interested in algebraic operations on rational behaviour,
i.\,e.,~behaviour of \emph{finite} coalgebras $(S, f)$ for a functor
$F$. Anticipating future applications in different categories than
$\set$, we present our results for endofunctors on general categories
$\A$ in which it makes sense to talk about ``finite'' objects and the
ensuing rational behaviour of ``finite'' coalgebras. So we work with
\emph{locally finitely presentable} categories of Gabriel and
Ulmer~\cite{gu71} (see also Ad\'amek and Rosick\'y~\cite{ar94}), and
we now briefly recall the basics.

A functor $F: \mathcal{A} \rightarrow \mathcal{B}$ is called
\emph{finitary} if $\mathcal{A}$ has and $F$ preserves filtered colimits. An object $X$ of a
category $\mathcal{A}$ is called \emph{finitely presentable} if its
hom-functor $\mathcal{A}(X, -)$ is finitary. 
A category $\mathcal{A}$ is \emph{locally finitely presentable} (lfp) if
(a)~it is cocomplete, and 
(b)~it has a set of finitely presentable objects such that every object
of $\mathcal{A}$ is a filtered colimit of objects from that set.
\begin{example}
  \begin{enumerate}[(1)]
  \item The category $\set$ and the categories of posets and graphs
    and their morphisms are lfp with finite sets, posets and graphs,
    respectively, as finitely presentable objects. 
  \item Finitary varieties are categories of algebras for a finitary
    signature satisfying a set of equations (e.\,g.,~groups, monoids,
    join-semilattices etc.). Such categories are lfp with the
    finitely presentable objects given by those algebras which can be
    presented by finitely many generators and relations. 
  \item As a special case consider \emph{locally finite} varieties,
    which are varieties where the free algebras on finitely many
    generators are finite (e.\,g.,~$\jsl$, distributive lattices or
    Boolean algebras). Here the finitely presentable objects are
    precisely the finite algebras. 
  \item Another special case of point~(2) are the categories $\Vec_{\F}$
    of vector spaces over a field $\F$, where the finitely presentable objects are precisely
    the finite dimensional vector spaces. 
  \end{enumerate}
\end{example}

\begin{remark}
  On the category $\set$, a finitary
  functor is determined by its behaviour on finite sets. More precisely,
  a functor $F: \set \rightarrow \set$ is finitary iff it is
  \emph{bounded} (see, e.\,g.,~Ad\'amek and Trnkov\'a~\cite{at}),
  i.\,e., for every set $X$ and every element $x \in FX$, there is a
  finite subset $i: Y\hookrightarrow X$ such that
  $x \in Fi[FY] \subseteq FX$. 
\end{remark}

\begin{example} We list some examples of finitary functors. 
  \begin{enumerate}[(1)]
  \item The \emph{finite} powerset functor $\powf$ is 
    finitary, whereas the ordinary powerset functor $\mathcal{P}$ is
    not.  
  \item The functor $FX = X^A$ is finitary if and only if $A$ is a
    finite set.
  \item More generally, the class of finitary set functors contains all
    constant functors and the identity functor, and it is closed under
    finite products, arbitrary coproducts and composition. Thus, a
    polynomial functor $\Sigma$ is finitary iff every operation symbol
    of the corresponding signature has finite arity (but there may be
    infinitely many operations).
  \item The functors $\mathcal{F}_{\M}$ are finitary for every monoid
    $\M$. 
  \item The functor $FX = \real \times X$ is finitary both on $\set$
    and on $\VecR$.
  \end{enumerate} 
\end{example}

\begin{assumption}
  Throughout the rest of this paper we assume, unless stated otherwise, that
  $\mathcal{A}$ is a locally finitely presentable category
  and $F: \mathcal{A} \rightarrow \mathcal{A}$ is a finitary functor. 
  So $F$ has a final coalgebra $t: \nu F \rightarrow F (\nu F)$ (see
  Makkai and Par\'e~\cite{mp89}). 
\end{assumption}

For a functor $F$ on an lfp category $\A$ the notion of a ``finite''
coalgebra is captured by a coalgebra having a finitely presentable
carrier. We denote by
\[
\coalgf(F)
\] the full subcategory of $F$-coalgebras $f: S \rightarrow FS$ with
$S$ finitely presentable. 
In order to talk about the behaviour of finite coalgebras in this
setting we would like to consider a coalgebra that is final among all
coalgebras in $\coalgf(F)$. However, $\coalgf(F)$ does not have a
final object in general, and so we consider the larger category of
locally finitely presentable coalgebras in which the desired final
object exists. 

An $F$-coalgebra $(S, f)$  is called \emph{locally finitely presentable} if
the canonical forgetful functor 
\[
\coalgf(F)/(S,f) \rightarrow \mathcal{A}/S
\] 
is cofinal~\cite{BMS12, Milius10}. In lieu of going into the details of this
definition we recall the following result, which gives a structure
theoretic characterisation of locally finitely
presentable coalgebras that we will use later:
\begin{theorem}[\cite{Milius10}]\label{thm-lfpc-colimit}
A coalgebra is locally finitely presentable iff it is a filtered colimit of 
a diagram of coalgebras from $\coalgf(F)$, i.\,e.,~a colimit of a
diagram of the form $\mathcal{D} \to \coalgf(F) \hookrightarrow \coalg(F)$.
\end{theorem}
\begin{example} We recall from~\cite{Milius10,BMS12} more concrete descriptions of locally
  finitely presentable coalgebras in some categories of interest. 
  \begin{enumerate}[(1)]
  \item A coalgebra for a functor on $\set$ is locally finitely
    presentable iff it is \emph{locally finite}, i.\,e., every finite subset of its carrier is
    contained in a finite subcoalgebra.

  \item Similarly, for a functor on a locally finite variety a
    coalgebra is locally finitely presentable iff every
    finite subalgebra of its carrier is contained in a finite
    subcoalgebra. 

  \item A coalgebra $(S,f)$ for a functor on $\Vec_\F$ is locally finitely
    presentable if and only if every finite dimensional subspace of
    its carrier $S$ is contained in a subcoalgebra $(S',f')$ of $(S,f)$
    whose carrier $S'$ is finite dimensional.
  \end{enumerate}
\end{example}

%More explicitly, $(S,s)$ is locally finitely
%presentable if and only if the following two conditions are satisfied:
%\begin{enumerate}
%\item for every $f: X \rightarrow S$ where $X$ is a finitely presentable object of $\mathcal{A}$
%there exists a coalgebra $(P, p)$ from $\coalgf(F)$, a coalgebra homomorphism $h:(P, p) \rightarrow (S,s)$
%and a morphism $f': X \rightarrow P$ such that $h \circ f'= f$, and
%\item the factorisation in (1) is essentially unique in the sense that for every
%$f'': X \rightarrow P$ with $h \circ f'' = f$ there exists a homomorphism $l: (P, p) \rightarrow (Q, q)$
%in $\mathit{Coalg_f}_f(F)$ and a coalgebra homomorphism $h': (Q, q) \rightarrow (S, s)$ such that $l\circ f' = l \circ f''$.
%\end{enumerate}

\subsection{The rational fixpoint}
The final $F$-coalgebra is thought to capture the behaviour of all
systems of type $F$. The behaviour of all ``finite'' systems is captured
by the so-called \emph{rational fixpoint}. We now recall its
definition and key properties as well as some illustrative examples
from~\cite{AMV06,Milius10,BMS12}. 

First it is easy to see that the category $\coalgf(F)$ is closed under finite
colimits, so the embedding
\begin{equation}\label{eq:E}
  E: \coalgf(F) \hookrightarrow \coalg(F)
\end{equation}
is an (essentially small) filtered diagram. We define
a coalgebra $$r: \rho F \rightarrow F (\rho F)$$
to be the colimit of $E$, i.e., $(\rho F, r) = \text{colim } E$. This
coalgebra is a fixpoint of $F$~\cite{AMV06}, and it is characterized
by a universal property both as a coalgebra and as an algebra. 
This is the content of the following theorem. Statement 3 in the theorem
below mentions iterative algebras for $F$. We do not recall that concept as it is
not needed in the present paper; we refer the interested reader to~\cite{AMV06}.
\begin{theorem}Let $(\rho F, r)$ be as above. Then 
\begin{enumerate}
\item $(\rho F, r)$ is a fixpoint of $F$, i.e., $r$ is an isomorphism, and
\item $(\rho F, r)$ is the final locally finitely presentable $F$-coalgebra, and finally
\item $(\rho F, r^{-1})$ is the initial iterative $F$-algebra.
\end{enumerate}
\end{theorem}

\begin{remark}\label{rem:subcoalg}
  For $\A = \set$ the rational fixpoint $\rho F$ is the union 
  of all images $f^{\dagger}[S] \subseteq \nu F$, where $f: S \rightarrow FS$
  ranges over the \emph{finite} $F$-coalgebras and $f^{\dagger}: S \to
  \nu F$ is the unique coalgebra homomorphism
  (see~\cite[Proposition~4.6 and Remark~4.3]{AMV06}).
  So, in particular, we see that $\rho F$ is a subcoalgebra of $\nu F$. 

  For endofunctors on different categories than $\set$, this need not be the case as shown
  in~\cite[Example~3.15]{BMS12}. However, for functors preserving
  monomorphisms on categories of vector spaces over a field and on locally
  finite varieties such as $\jsl$  the rational fixpoint always is a
  subcoalgebra of $\nu F$ (see~\cite[Proposition~3.12]{BMS12}). 
\end{remark}

\begin{example}\label{ex:rat}
For each of the functors in Example~\ref{ex:coalg} we now mention the
rational fixpoints. For more examples see~\cite{AMV06,BMS12}.
\begin{enumerate}[(1)]
\item For the functor $FX = \real \times X$ on $\set$ whose final
  coalgebra is carried by the set of all streams over $\real$, the
  rational fixpoint consists of all streams that are \emph{eventually
    periodic}, i.e., of the form $\sigma = vwwww\ldots$ for words $v
  \in \real^*$ and $w \in \real^+$.  If we consider the similar functor $FV = \real \times V$
  on the category of vector spaces over $\real$, the rational fixpoint
  consists precisely of all \emph{rational streams} (see, 
  e.\,g.,~Rutten~\cite{Rutten08}).

\item Recall that deterministic automata are modeled by the functor
  $FX = 2 \times X^A$ on $\set$.  The carrier of the rational fixpoint
  of $F$ is the set of all languages accepted by \emph{finite}
  automata, viz.~the set of all \emph{regular} languages.
  If we define $F$ instead on the category $\jsl$ of
  join-semilattices, its rational fixpoint is still given by all
  regular languages, this time with the join-semilattice structure
  given by union and $\emptyset$.

\item For $FX = \powf(A \times X)$ on $\set$ we saw in
  Example~\ref{ex:coalg}(3) that the
  coalgebras are labelled transition systems and $\nu F$ consists of
  processes (modulo strong bisimilarity). In this case the rational
  fixpoint contains all \emph{finite-state} processes (modulo
  bisimilarity); more precisely, $\rho F$ is the coproduct of all
  \emph{finite} $F$-coalgebras modulo the largest bisimulation---this
  follows from the construction of $\rho F$ as the colimit of the
  diagram in~\refeq{eq:E}. 

\item Similarly, for $FX = 2 \times (\powf X)^A$ on $\set$, $\rho F$
  can be described as the coproduct of all finite $F$-coalgebras
  modulo the largest bisimulation. A different (isomorphic)
  description is that $\rho F$ consists of all \emph{rational}
  finitely branching strongly extensional trees with edges labelled in
  $A$ and nodes labelled in $2$, where a tree is rational if it has
  (up to isomorphism) only a finite number of subtrees.

\item For the functor $FX = (\mathcal{F}_{\M} X)^A$ of weighted transition
  systems the rational fixpoint is obtained as the coproduct of all finite
  WTS's modulo weighted bisimilarity. %observational equivalence. 

\takeout{ % example not treated in the paper
\item Consider, for a field $\mathbb{F}$ (for example
  $\real$), the functor $FV = \mathbb{F} \times V^A$ on $\Vec_\F$
  vector spaces over $\mathbb{F}$. The rational fixpoint of $F$ then
  consists of all rational formal power series~\cite{BMS12}. 
}

\item Let $F = \Sigma$ be a polynomial functor on $\set$, where the
  final coalgebra is carried by all $\Sigma$-trees. Then the rational
  fixpoint is given by all \emph{regular} $\Sigma$-trees (see
  Courcelle~\cite{Courcelle83}), i.\,e., all those $\Sigma$-trees
  having (up to isomorphism) only finitely many different subtrees;
  this description of regular trees is due to Ginali~\cite{ginali}.

%(to do: non-deterministic automata, weighted automata)
%\smnote{I'd do LTS's instead of non-det. autom.}
\end{enumerate}
\end{example}

\section{Bipointed specifications}\label{sec-sos}

We still assume that $F: \A \to \A$ is a finitary endofunctor on the
lfp category $\A$. 

\begin{definition}\label{def-bipointed}
Let $\Sigma: \mathcal{A} \rightarrow \mathcal{A}$ be a functor. We call
a  natural transformation 
$$\lambda: \Sigma(F \times \id) \Rightarrow F(\Sigma + \id)$$
a \emph{bipointed specification}. 
\end{definition}
While this is a rather abstract and seemingly unusable specification format, by 
considering a specific functor $F$ one can often devise more concrete formats. We
discuss several examples in Section~\ref{sec-app}. 
For now let us consider the
definition of a parallel operator on transition systems, to give a basic
example of a bipointed specification. Klin~\cite[\S5.2]{Klin07} presents
a similar example and notices that it gives rise to a bipointed specification.
\begin{example}\label{ex-lts-parallel}
Recall that the functor corresponding to transition systems is $FX = \mathcal{P}_f(A \times X)$
on $\set$ and that we think of the elements of $\nu F$ as processes. 

We would like to define a parallel operator on processes, which can be defined
in standard SOS as follows:
$$
\frac{s \stackrel{a}{\rightarrow} s'}
{s || t \stackrel{a}{\rightarrow} s' || t}
\qquad\qquad
\frac{t \stackrel{a}{\rightarrow} t'}
{s || t \stackrel{a}{\rightarrow} {s || t'}}
$$
Intuitively this means that whenever $s$ can make an $a$-transition to some state $s'$, then
$s || t$ can make an $a$-transition to $s' || t$, and similarly for $t$.
Since we are interested in a single binary operator, the corresponding signature is 
$\Sigma X = X \times X$. Thus, the bipointed specification $\lambda:
\Sigma (F \times \id) \Rightarrow F(\Sigma + \id)$ is given by the following
family of maps: 
$$
\lambda_X: (\mathcal{P}_f(A \times X) \times X) \times (\mathcal{P}_f(A \times X) \times X) \rightarrow 
\mathcal{P}_f(A \times (X \times X + X)).
$$ 
Now a for a 4-tuple $(S,s,T,t)$ in the domain of $\lambda_X$, $S$ and
$T$ are the sets of outgoing transitions of $s$ and $t$,
respectively. Moreover, an element $(a, (u, v))$ in the codomain of
$\lambda_X$ corresponds to an $a$-transition to the state $u||v$. 
Thus, we may define $\lambda_X$ as
$$\lambda_X(S,s,T,t) = \{(a, (s',t)) \mid (a,s') \in S\} \cup \{(a,(s,t')) \mid (a,t') \in T\}.$$
\end{example}
It has been shown by Turi and Plotkin~\cite{TP97} and Bartels~\cite{Bartels04} that
natural transformations as in the previous definition and more general
ones (see Klin~\cite{Klin11} for an overview) induce algebraic structures on the final coalgebra $\nu F$. We 
recall how this construction works for our bipointed specifications. To this
end let $\lambda: \Sigma(F \times \id) \rightarrow F(\Sigma + \id)$ be a bipointed specification. 
We define a functor $\Phi: \coalg(F) \rightarrow \coalg(F)$ as
follows:
\begin{equation}\label{eq:phi}
\begin{array}{rcl}
  \Phi(S, f) & = &\left(\Sigma S + S \xrightarrow{\Sigma \langle f, id \rangle + f}
    \Sigma(F S \times S) + FS \xrightarrow{[\lambda_{S} , F \inr]}
    F(\Sigma S + S)\right), \\
  \Phi h & = & \Sigma h + h, \qquad \text{for any coalgebra homomorphism
  $h: (S, f) \to (T, g)$.}
\end{array}
\end{equation}
In order for $\Phi$ to be well-defined $\Phi h$ must be a coalgebra homomorphism, which indeed
follows from naturality of $\lambda$ and functoriality of $\Sigma$. We do not spell out the details,
but refer the interested reader to~\cite{Bartels04,Klin11}. Observe
that $\Phi$ is a \emph{lifting} of $\Sigma + \id$ to $\coalg(F)$,
i.\,e., for the forgetful functor $U: \coalg(F) \to \set$ we have
$(\Sigma + \id) \cdot U = U \cdot \Phi$.
\takeout{ % to save space
the commutative square below:
\begin{equation}\label{diag:lift}
  \vcenter{
    \xymatrix{
      \coalg(F)
      \ar[r]^-\Phi
      \ar[d]_U
      &
      \coalg(F)
      \ar[d]^{U}
      \\
      \set 
      \ar[r]_-{\Sigma + \id}
      &
      \set
    }
  }
\end{equation}}

Now if we apply $\Phi$ to the final coalgebra $(\nu F, t)$ we obtain the following:
$$
\Sigma (\nu F) + \nu F \xrightarrow{\Sigma \langle t, id \rangle + t}
\Sigma(F (\nu F) \times \nu F) + F (\nu F) \xrightarrow{[\lambda_{\nu F} , F \inr]}
F(\Sigma (\nu F) + \nu F).
$$
By finality, there is a unique coalgebra homomorphism from $\Phi(\nu
F, t)$ to $(\nu F, t)$, and it is easy to prove that its right-hand
component is the identity on $\nu F$; so the homomorphism has the form
$$[\alpha, id]: \Sigma (\nu F) + \nu F \rightarrow \nu F.$$
Thus, we obtain a unique $\Sigma$-algebra $\alpha: \Sigma \nu F \rightarrow \nu F$
making the diagram below commute: 
\begin{equation}
  \label{diag:alpha}
  \vcenter{
    \xymatrix@C+1pc{
      \Sigma (\nu F) \ar[r]^-{\Sigma\langle t, id \rangle} \ar[d]_{\alpha}
      & 
      \Sigma(F (\nu F) \times \nu F) \ar[r]^-{\lambda_{\nu F}} 
      & 
      F(\Sigma (\nu F) + \nu F) \ar[d]^{F[\alpha, id]}
      \\
      \nu F \ar[rr]^t & & F \nu F
    }
  }
\end{equation}
In concrete instances, $\alpha$ provides the denotational semantics of
the algebraic operations as specified by $\lambda$, taking as 
arguments elements of the final coalgebra. Returning to the above Example~\ref{ex-lts-parallel},
for two processes $s$ and $t$, $\alpha(s,t)$ is indeed the parallel
composition $s||t$.

\begin{remark}
  \label{rem:GSOS}
The original abstract GSOS format considered by Turi and Plotkin is given by
natural transformations of the form 
$$\lambda: \Sigma (F \times \id) \Rightarrow FT_{\Sigma}$$
where $T_{\Sigma}$ is the free monad on $\Sigma$; for a polynomial
functor $\Sigma$ on $\set$, $T_{\Sigma}X$ is the set of
all terms of operations in $\Sigma$ over variables of $X$. This is more general
than the bipointed specifications of Definition~\ref{def-bipointed}. However, 
we will be interested in operations on the \emph{rational}
fixpoint. And in general,
operations on $\nu F$ defined by the above format need not restrict to
$\rho F$ as demonstrated by the following example.
\end{remark}
\begin{example}\label{ex-gsos-op}
Recall from Example~\ref{ex:rat}(1) the functor $FX = \real \times X$ whose coalgebras are stream
systems. A unary operation $p$ on the final coalgebra $\nu F = \real^\omega$ of
all real streams is specified by the following behavioural
differential equations:
\[
p(\sigma)(0) = \sigma(0) + 1
\qquad
p(\sigma)' = p (p(\sigma')),
\]
where $\sigma' = (\sigma(1), \sigma(2), \sigma(3), \ldots)$ denotes
the tail of the stream $\sigma$. Let $\Sigma X = X$ be the polynomial
functor for the signature with one unary operation symbol $p$. Then
the above behavioural differential equations give rise to the natural
transformation
\[
\ell_X: \Sigma FX = \real \times X \to \real \times T_\Sigma X =
FT_\Sigma X
\qquad
(r, x) \mapsto (r+1, p(p(x))),
\]
and we get an abstract GSOS rule as follows:
$
\lambda = (\xymatrix@1{
\Sigma(F \times \id) \ar@{=>}[r]^-{\Sigma \pi_0} & \Sigma F \ar@{=>}[r]^-{\ell} &
F T_\Sigma
}),
$
where $\pi_0: F \times \id \Rightarrow F$ denotes the left-hand product projection.
It is easy to see that the ensuing operation $p: \nu F \to \nu F$
satisfies
\[
(0,0,0, \ldots) \stackrel{p}{\longmapsto} (1, 2, 4, 8, \ldots, 2^n, \ldots).
\]
Clearly, the rational fixpoint $\rho F$, which consists of eventually
periodic streams, is not closed under the operation $p$. 
\end{example}

Even operations defined using bipointed specifications will not
restrict to $\rho F$ in general, when we simultaneously specify
infinitely many operations that depend on one another.

\begin{example}
  For $FX = \real \times X$ on $\set$ with $\nu F = \real^\omega$ we
  define infinitely many unary operations $u_n$, $n \in \nat$, by the following
  behavioural differential equations:
  \[
  u_n(\sigma)(0) = n 
  \qquad
  u_n(\sigma)' = u_{n+1}(\sigma').
  \]
  Let $\Sigma X = \nat \times X$ be the polynomial functor
  corresponding to the signature with the unary operation symbols
  $u_n$, $n \in\nat$. Then the above behavioral differential equations
  give rise to the natural transformation
  \[
  \ell: \Sigma FX = \nat \times \real \times X \to \real \times \nat
  \times X = F \Sigma X
  \qquad
  (n, r, x) \mapsto (n, n+1, x),
  \]
  and we get a bipointed specification as follows:
  $
  \lambda = (\xymatrix@1{
    \Sigma(F \times \id) \ar@{=>}[r]^-{\Sigma \pi_0} & \Sigma F
    \ar@{=>}[r]^-{\ell} &
    F\Sigma \ar@{=>}[r]^-{F\inl} & F(\Sigma + \id).
  }
  $
  The ensuing operations $u_n: \nu F \to \nu F$ satisfy 
  $
  (0,0,0,\ldots) \stackrel{u_n}{\longmapsto} (n,n+1,n+2,n+3,\ldots).
  $
  So the rational fixpoint $\rho F$ is not closed under these
  operations. 
\end{example}

\section{Algebras on the rational fixpoint}\label{sec-alg} 

In this section we show how a bipointed specification defines an
algebraic structure $\beta: \Sigma(\rho F) \to \rho F$ on the rational fixpoint similar to
the structure $\alpha: \Sigma(\nu F)\to \nu F$ in~\refeq{diag:alpha}. We will also see that
the new structure $\beta$ on $\rho F$ is a ``restriction'' of
$\alpha$; more precisely the unique coalgebra homomorphism $(\rho F,
r) \to (\nu F, t)$ is also a $\Sigma$-algebra homomorphism. 
In order to proceed we make
\begin{assumption}~\label{ass-sig} We still assume that $F$
  is a finitary functor on the lfp category $\A$. We now assume also
  that $\Sigma: \mathcal{A} \rightarrow \mathcal{A}$ is a 
  \emph{strongly finitary} functor, i.\,e., $\Sigma$ is finitary and it
  preserves finitely presentable objects. We also assume 
  that $\lambda: \Sigma(F \times \id) \rightarrow F(\Sigma + \id)$ is a
  bipointed specification.
  We still write $\Phi$ for the functor in~\refeq{eq:phi}, which lifts
  $\Sigma + \id$ to $\coalg(F)$. 
\end{assumption}

\begin{example} The notion of strongly finitary functor is taken
  from~\cite{AMV03} and we discuss some examples below. 
  \begin{enumerate}[(1)]
  \item The class of strongly finitary functors on $\set$ contains the
    identity functor, all constant functors on finite sets, the finite
    power-set functor $\powf$, and it is closed
    under finite products, finite coproducts and composition. 

  \item From the previous point we see that a polynomial functor
    $\Sigma$ on $\set$ is strongly finitary iff the corresponding
    signature has finitely many operation symbols of 
    finite arity.

  \item The functor $FX = 2 \times X^A$ is strongly finitary iff $A$
    is a finite set. 

  \item The type functor $FX = \real \times X$ of stream systems as
    coalgebras is finitary but not strongly so. However, if we
    consider $F$ as a functor on $\VecR$, then it is strongly
    finitary; in fact, for every finite dimensional real vector space
    $X$, $\real \times X$ is finite dimensional, too.
  \end{enumerate}
\end{example}

%We will construct a coalgebra with carrier $\Sigma(\rho F) + \rho F$ in a similar way
%as we have done so in Section~\ref{sec-sos} for the final coalgebra $\nu F$. But
%first we need the following lemma, which states that $\Phi$ is well-defined as a
%lifting of $\Sigma + \id$ to the category $\coalgf(F)$ of coalgebras
%with a finitely presentable carrier. 
First we need the following lemma which states that $\Phi$ is a
finitary functor that restricts to the subcategory of coalgebras with
a finitely presentable carrier. 
\begin{lemma}\label{lm-lift-restr}
The lifting $\Phi$ (a)~is finitary and (b)~restricts to $\coalgf(F)$. 
\end{lemma}
\begin{proof}
  Ad (a). By assumption, $\Sigma$ is a finitary functor, and so
  $\Sigma + \id: \mathcal{A} \rightarrow \mathcal{A}$ is clearly
  finitary, too. Since the forgetful functor
  $U: \coalg(F) \rightarrow \mathcal{A}$ creates all colimits, it
  follows that $\Phi$ is finitary since $(\Sigma + \id) \cdot U = U
  \cdot \Phi$. 
  % (cf.~Diagram~\refeq{diag:lift}). 

Ad~(b). Let $(S, f)$ be an object of $\coalgf(F)$.
Then $S$ is finitely presentable, and, since $\Sigma$ is strongly
finitary, $\Sigma S$ is also finitely presentable. Finally, since finitely presentable
objects are clearly closed under finite colimits, $\Sigma S + S$ is finitely
presentable, too. Thus, $\Phi (S, f)$ is an object of $\coalgf(F)$.
\end{proof}
Now in order to use the universal property of $\rho F$ we prove that
the lifting $\Phi$ applied to it is locally finitely presentable:
\begin{lemma}\label{lm-rfp}
The coalgebra $\Phi(\rho F, r)$ is locally finitely presentable.
\end{lemma}
\begin{proof}
  Since $(\rho F, r) = \text{colim } E$ (see~\refeq{eq:E}) and $\Phi$ is
  finitary (Lemma~\ref{lm-lift-restr}(a)), 
  $\Phi(\rho F, r)$ can be obtained as the filtered colimit of the diagram
  $
  \coalgf(F) \stackrel{E}{\longhookrightarrow} 
  \coalg(F) \stackrel{\Phi}{\longrightarrow}
  \coalg(F).
  $
  By Lemma~\ref{lm-lift-restr}(b), this is a diagram of coalgebras
  from $\coalgf(F)$. Therefore, by Theorem~\ref{thm-lfpc-colimit}, $\Phi(\rho F, r)$
  is a locally finitely presentable coalgebra.
\end{proof}
From the above lemma, by the universal property of the rational
fixpoint we obtain
\begin{corollary}\label{cor-alg-rational}
There exists a unique algebra structure $\beta: \Sigma(\rho F) \rightarrow \rho F$
such that the following diagram commutes:
$$
\xymatrix@C+2pc{
\Sigma (\rho F) \ar[r]^-{\Sigma\langle r, id \rangle} \ar[d]_{\beta}
& \Sigma(F (\rho F) \times \rho F) \ar[r]^{\lambda_{\rho F}} 
& F(\Sigma (\rho F) + \rho F) \ar[d]^{F[\beta, id]}\\
\rho F \ar[rr]^r & & F(\rho F)
}
$$
\end{corollary}
Indeed, by Lemma~\ref{lm-rfp} and the finality of $\rho F$ as a locally finitely presentable coalgebra
there is a unique coalgebra homomorphism from $\Phi(\rho F,r)$ to $(\rho F,r)$, and it is again easy
to show that its right-hand coproduct component must be the identity, and so its left-hand component is
the desired $\Sigma$-algebra structure $\beta$.

% \smnote{Already said before.}
%We conclude this section with a proposition stating that the
%algebra structure $\beta: \Sigma(\rho F) \to \rho F$ is a
%``restriction'' of the algebra structure $\alpha: \Sigma(\nu F) \to
%\nu F$ from Diagram~\refeq{diag:alpha} to $\rho F$. In other words we
%obtain our desired result that $\rho F$ is closed under the operations
%on $\nu F$ specified by the given bipointed specification $\lambda$. 

\begin{proposition}\label{prop:hom}
  Let $h: (\rho F, r) \to (\nu F, t)$ be the unique $F$-coalgebra
  homomorphism. Then $h$ is also a $\Sigma$-algebra homomorphism from
  $(\rho F, \beta)$ to $(\nu F, \alpha)$. 
\end{proposition}
\begin{proof}
  We are to prove the equation $h \o \beta = \alpha \o \Sigma h$. This is
  equivalent to proving 
  \[
  [h \o \beta, h] = [\alpha \o \Sigma h, h]: \Sigma(\rho F) + \rho F \to \nu F,
  \]
  which is established by proving that both sides form coalgebra
  homomorphisms from $\Phi(\rho F, r)$ to $(\nu F, t)$. 
  Indeed, they are both compositions of two coalgebra homomorphisms:
  \[
  [h \o \beta, h] = (\xymatrix@1{
    \Phi(\rho F, r) \ar[r]^-{[\beta, id]} & (\rho F, r) \ar[r]^-h &
    (\nu F, t)
    }), \qquad
  [\alpha \o \Sigma h, h] = (\xymatrix@1{
    \Phi(\rho F, r) \ar[r]^-{\Phi h} & \Phi(\nu F, t)
      \ar[r]^-{[\alpha, id]} & (\nu F, t)
    }).\qedhere
  \]
\end{proof}
As a consequence we obtain the following closure property of $\rho F$:
Suppose that $h$ in the previous proposition is a monomorphism (cf.~Remark~\ref{rem:subcoalg}). Then 
$(\rho F, r)$ is a subcoalgebra of $(\nu F, t)$ and $(\rho F, \beta)$
is a subalgebra of $(\nu F, \alpha)$ via $h$.  

\begin{remark}
  Notice that the results of this section are easily seen to
  generalize from bipointed specifications to the more general
  \emph{coGSOS laws}, i.\,e., natural transformations of the form
  \[
  \lambda: \Sigma C_F \to F(\Sigma + \id),
  \]
  where $C_F$ denotes the cofree comonad on $F$ (see,
  e.\,g.,~\cite{Klin11}). (Observe that the cofree comonad on $F$ is
  given objectwise by assigning to an object $X$ of $\A$ the final
  coalgebra $\nu(F(-)\times X )$.) This is formally dual to the 
  abstract GSOS format we recalled in
  Remark~\ref{rem:GSOS}. coGSOS laws allow to specify important
  operations not captured by bipointed specifications, e.\,g.,~the
  tail operation $\sigma \mapsto \sigma'$ on streams. And in the case of transition system
  specifications (i.\,e., where $FX = \powf(A \times X)$) it is well-known
  that specifications in the so-called safe ntree format are 
  instances of coGSOS laws (see~\cite{TP97}), but it is not known whether every
  coGSOS law arises from a safe ntree specification. We defer a thorough
  treatment of coGSOS laws to future work. 
  
  \takeout{ % I took this out again - seemed to long. Stefan.
  Here we have restricted ourselves to working with bipointed
  specifications because in many concrete application
  this format is sufficient to specify the desired operations and
  because bipointed specifications closely correspond to transition
  system specifications in the simple GSOS format
  of~\cite{AFV}. However, coGSOS laws do allow to specify important
  operations not captured by bipointed specifications, e.\,g., the
  tail operation $\sigma \mapsto \sigma'$ on streams. We defer a thorough
  treatment of coGSOS laws to future work. } % end takeout
\end{remark}

\section{Applications}\label{sec-app}

In this section we consider algebraic operations defined on the rational fixpoint for several
concrete types of systems, as applications of Corollary~\ref{cor-alg-rational}
and Proposition~\ref{prop:hom}. We discuss concrete SOS formats corresponding to bipointed
specifications. There are many such concrete specification formats for similar
distributive laws studied in the literature~\cite{Klin11}, and we can only cover a few examples here. 
For most of these formats it is easy to obtain a restriction to bipointed specifications, so
that our results apply and the obtained specifications define operations which restrict to the
rational fixpoint. Throughout this section we assume that $\Sigma$ is
a signature represented as a strongly finitary polynomial functor on
$\set$. 
%\smnote{Added sentence for referee \#3.}
To the best of our knowledge, all the results we present in the
corollaries in this section are new. 

\vspace*{-10pt}
\paragraph{Streams.}
Consider the $\set$ functor $FX = \real \times X$ of streams over the
reals. A bipointed specification then is a natural transformation $\lambda$ with components
\begin{equation}\label{eqn-lambda-streams}
\lambda_X: \Sigma(\real \times X \times X) \Rightarrow \real \times (\Sigma X + X).
\end{equation}
We recall from~\cite{Klin11} that these  natural transformations can be expressed in a more convenient SOS format as follows.
A \emph{bipointed stream SOS rule} for an operator $f$ in $\Sigma$ of arity $n$ is a rule
$$
\frac{x_1 \stackrel{r_1}{\rightarrow} x_1' \qquad \ldots \qquad x_n \stackrel{r_n}{\rightarrow} x_n'}
{f(x_1, \ldots, x_n) \stackrel{r}{\rightarrow} t}
$$
where $x_1, \ldots, x_n, x_1', \ldots, x_n'$ is a collection of
pairwise distinct variables, which we call $V$. Further, $t$ is a
variable in V or a term of the form $g(y_1, \ldots, y_m)$ where $g$ is
an $m$-ary operation symbol of $\Sigma$, and $y_i \in V$ for all $1 \leq i \leq m$, and finally $r, r_1, \ldots, r_n \in \mathbb{R}$.
\takeout{
\begin{itemize}
  \item $x_1, \ldots, x_n, x_1', \ldots, x_n'$ is a collection of pairwise distinct variables, which we call $V$
  \item $t$ is a variable in V or a term of the form $g(y_1, \ldots, y_m)$ where $g$ is an operation
  in $\Sigma$ of arity $m$, and $y_i \in V$ for all $1 \leq i \leq m$,
  \item $r, r_1, \ldots, r_n \in \mathbb{R}$
\end{itemize}
}%
We say the above rule is \emph{triggered} by the $n$-tuple $(r_1, \ldots, r_n)$.
A \emph{bipointed stream SOS specification} for the strongly finitary signature $\Sigma$ then is a collection of bipointed stream
SOS rules for $\Sigma$ such that for each operator $f$ in $\Sigma$ and for each sequence of real numbers $r_1, \ldots, r_n$, 
there exists precisely one rule for $f$ triggered by $(r_1, \ldots, r_n)$. Bipointed stream SOS specifications are in
one-to-one correspondence with natural transformations of the above type (\ref{eqn-lambda-streams}). Therefore, by
Proposition~\ref{prop:hom} we have
\begin{corollary}
The operations defined by a bipointed stream SOS specification on the final coalgebra of the $\set$ functor $FX = \mathbb{R} \times X$ restrict to the
rational fixpoint of $F$, i.e., the coalgebra of eventually periodic streams.
\end{corollary}

As an example consider the well-known \emph{zip} (or \emph{merge}) operation, 
which takes two streams and returns a new stream which alternates between the two
given arguments. The standard definition of \emph{zip} can be given as a bipointed
stream SOS rule:
$$
\frac{\sigma \stackrel{r_1}{\rightarrow} \sigma' ~~~~ \tau \stackrel{r_2}{\rightarrow} \tau'}
{\mathit{zip}(\sigma, \tau) \stackrel{r_1}{\rightarrow} \mathit{zip}(\tau, \sigma')}
$$
A direct consequence of the above corollary is the basic insight that for any two 
streams $\sigma$ and $\tau$ which are eventually periodic, $\mathit{zip}(\sigma, \tau)$ 
is again eventually periodic.

\begin{remark}
  \begin{enumerate}[(1)]
  \item Another way of specifying operations on streams is using
    \emph{behavioural differential equations}~\cite{Rutten05}
    (cf.~Example~\ref{ex-gsos-op}). In fact the above bipointed stream
    specifications also correspond precisely to behavioural
    differential equations in which each of the derivatives is
    restricted to be either a variable or a single operator applied to
    variables (precisely as $t$ in the definition of bipointed stream
    SOS rules).  Thus, such differential equations define operations
    which restrict to eventually periodic streams as well.
  \item If we consider $FX = \real \times X$ as a functor on $\VecR$
    then bipointed specifications are natural transformations $\lambda$
    where $\Sigma$ is a functor on $\VecR$ and where the components
    $\lambda_X$ in~\refeq{eqn-lambda-streams} are linear maps. By
    Proposition~\ref{prop:hom} we obtain that operations defined by a
    bipointed specification on $\nu F$, the final coalgebra of all
    streams, restrict to the rational fixpoint $\rho F$ formed by all
    rational streams. An example of such an operation is
    the above specification of $zip$. Consequently, we obtain that
    rational streams are closed under $zip$. 
  \end{enumerate}
\end{remark}

\vspace*{-10pt}
\paragraph{Labelled transition systems.}
Recall from Example~\ref{ex:coalg}(3) that labelled transition systems are coalgebras for the functor 
$FX=\powf(A \times X)$ on $\set$. %We assume that $A$ is a finite set. \jrnote{why should A be finite?}
%For labelled transition systems, the format instantiates to
In this case a bipointed specification for a strongly finitary signature $\Sigma$ is a natural transformation with components
\begin{equation}\label{eq:lts:dl}
\lambda_X: \Sigma(\mathcal{P}_f(A \times X) \times X) \Rightarrow
\mathcal{P}_f(A \times (\Sigma X + X)).
\end{equation}
This corresponds to a restricted ``flat'' version of the well-known
GSOS format~\cite{BIM95}, where on the right-hand side of the
transition in the conclusions of a rule there may only be a variable
or single operation symbol applied to variables in lieu of an
arbitrary term. %
%The general GSOS format as given in~\cite{Klin11} (it was first
%formulated in~\cite{BIM95}) is easy to restrict to bipointed
%specifications. In order to do so 
For a strongly finitary signature, this is precisely the simple GSOS
format of~\cite{AFV}.
Indeed, following the presentation in~\cite{Klin11}, we define a \emph{bipointed LTS SOS rule} for an operator $f$ in $\Sigma$ of arity $n$ as 
\begin{equation}
\frac{\{x_{i_j} \stackrel{a_j}{\rightarrow} y_j\}_{j = 1..m} \qquad \{x_{i_k} \stackrel{b_k}{\not \rightarrow}\}_{k = 1..l}}
{f(x_1, \ldots, x_n) \stackrel{c}{\rightarrow} t}
\end{equation}
where $m$ is the number of positive premises and $l$ is the number of negative premises. The variables $x_1, \ldots, x_n, y_1, \ldots, y_m$
are again pairwise distinct; let $V$ denote the set of these
variables. Then $t$ is either a variable in $V$ or a flat term
$g(z_1, \ldots, z_p)$, where $g$ is an $p$-ary operation symbol in
$\Sigma$ and $z_1, \ldots, z_p \in V$. 
%\smnote{I corrected ``terms'' to ``\emph{flat} terms''.}%
Finally $a_1, \ldots, a_m, b_1, \ldots, b_l, c \in A$ are labels. The above rule
is \emph{triggered} by an $n$-tuple $(E_1, \ldots, E_n)$, where each
$E_i \subseteq A$, 
if for each $i = 1..n$ we have $a_j \in E_{i_j}$ for all $j = 1..m$ and $b_k \not \in E_{i_k}$ for all $k=1..l$. A 
\emph{bipointed LTS SOS specification} then is a collection of rules of the above
type such that for each operator $f$ in $\Sigma$, each $c \in A$ and
each $n$-tuple $\bar{E} = (E_1, \ldots, E_n)$ of sets of labels, 
there are finitely many rules for $f$ with $c$ as the conclusion label that are triggered  by $\bar{E}$. 
Bipointed specifications for labelled transition systems (\ref{eq:lts:dl}) are in one-to-one correspondence with bipointed LTS SOS specifications.
So by Proposition~\ref{prop:hom} we have
\begin{corollary}\label{cor-lts-closure1}
The operations defined by a bipointed LTS SOS specification on the final coalgebra of the $\set$ functor $FX = \powf(A \times X)$ restrict to the
rational fixpoint of $F$, i.e.,~the coalgebra of all \emph{finite}
labelled transition systems modulo the largest bisimulation.
\end{corollary}
%Indeed, the follows from Proposition~\ref{prop:hom} using that the two
%coalgebras mentioned above form the rational fixpoint of $FX = \real
%\times X$ considered as a functor on $\set$ and on $\VecR$, respectively.

%Let $A$ be a finite set of actions. We will consider Milner's CCS~\cite{milner}, which has the following syntax:
%$$P ~::=~0 ~|~ a.P ~|~ P_1 + P_2 ~|~ P_1 || P_2 ~|~ P[\oldrho] ~|~P\setminus L$$
%where $a$ ranges over $A$ and $L$ over $\pow(A)$.
%The semantics of the operators fall within our format:
As an example we recall the semantics of the operators of Milner's CCS~\cite{milner}, which forms
a bipointed LTS SOS specification:
$$
\frac{}{a.P \stackrel{a}{\rightarrow} P}
~~~~
\frac{P_1 \stackrel{a}{\rightarrow} P_1'}
{P_1 + P_2 \stackrel{a}{\rightarrow} P_1'}
~~~~
\frac{P_2 \stackrel{a}{\rightarrow} P_2'}
{P_1 + P_2 \stackrel{a}{\rightarrow} P_2'}
%~~~~
%\frac{P_i \stackrel{a}{\rightarrow} P_i'}
%{P_1 + P_2 \stackrel{a}{\rightarrow} P_i'}~~~i \in \{1,2\}
~~~~
\frac{P \stackrel{a}{\rightarrow} P' }
{P \setminus L \stackrel{a}{\rightarrow} P' \setminus L}
(a, \bar{a} \not \in L)
$$
$$
~~~~
\frac{P_1 \stackrel{a}{\rightarrow} P_1'}
{P_1 || P_2 \stackrel{a}{\rightarrow} P_1' || P_2}
~~~~
\frac{P_2 \stackrel{a}{\rightarrow} P_2'}
{P_1 || P_2 \stackrel{a}{\rightarrow} {P_1 || P_2'}}
~~~~
\frac{P_1 \stackrel{a}{\rightarrow} P_1' ~~~ P_2 \stackrel{\bar{a}}{\rightarrow} P_2'}
{P_1 || P_2 \stackrel{\tau}{\rightarrow} {P_1' || P_2'}}
~~~~
\frac{P \stackrel{a}{\rightarrow} P'}
{P[\oldrho] \stackrel{\oldrho(a)}{\rightarrow} P'[\oldrho]}
$$
Note that in order for the signature corresponding to these operations to be strongly finitary, 
the set of actions $A$ must be finite. Then, by the above
Corollary~\ref{cor-lts-closure1}, finite-state processes are closed
under all of the above operations.
%applying any of the above operations on finite transition systems yields again a finite transition system.

\begin{remark}
  Aceto~\cite{aceto94} proved (see~\cite[Theorem~5.28]{AFV}) that for
  a simple GSOS specification the induced transition system on the
  process terms is regular, i.\,e., for every closed process term $P$
  the transition system giving $P$ its operational semantics has
  finitely many states. Note that this result is not a direct
  consequence of our results in Section~\ref{sec-alg}. In fact, the
  transition systems induced by a (simple) GSOS specification is
  (generalized by) the \emph{operational model} of Turi and
  Plotkin~\cite{TP97} for the corresponding abstract GSOS
  specification; this operational model is the initial
  $\Sigma$-algebra $\mu \Sigma$ equipped with the $F$-coalgebra
  structure induced by the abstract GSOS specification. The
  corresponding generalization of Aceto's result then states that for
  a bipointed specification $\lambda$ the induced 
  $F$-coalgebra on $\mu \Sigma$ is locally finitely presentable. We shall
  state and prove this result in future work.
\end{remark}

\vspace*{-10pt}
\paragraph{Non-deterministic automata.}
Recall from Example~\ref{ex:coalg}(4) that non-deterministic automata are coalgebras for the $\set$ functor $FX = 2\times (\powf X)^A$. 
Bipointed specifications for this functor instantiate to natural transformations with components
\begin{equation}\label{eq:nondet}
  \lambda_X: \Sigma(2 \times \mathcal{P}_f(X)^A \times X) \Rightarrow 2
  \times \mathcal{P}_f(\Sigma X + X)^A.
\end{equation}
We are not aware of an existing SOS format for non-deterministic automata corresponding precisely to these 
natural transformations, which we call \emph{bipointed NDA specifications}. However, it is not hard to devise 
a format based on the above LTS SOS specifications, such that each specification gives rise to a bipointed
NDA specification, but not necessarily vice versa, i.e.,~an incomplete format. 
Define an \emph{output rule}
for an operator $f$ in $\Sigma$ of arity $n$ as
\begin{equation}
\frac{\{x_{i_j}\downarrow\}_{j=1..k}}
{f(x_1, \ldots, x_n)\downarrow}
\end{equation}
%\smnote{I made a change here.}%
where $k \leq n$. The above output rule is \emph{triggered} by an
$n$-tuple $(o_1, \ldots, o_n) \in 2^n$ provided that for all $j$, $o_{i_j} = 1$ iff
$x_{i_j} \downarrow$ is in the premise of the rule.
Intuitively, such a rule specifies that $f(x_1, \ldots, x_n)\downarrow$, meaning that 
$f(x_1, \ldots, x_n)$ is a final state, whenever each of its arguments 
$x_{i_j}$ are final, and all of the other arguments are not final. Notice that
one way to extend this format would be to make the transitions also depend on the
output of the arguments; for technical convenience and lack of space we do not discuss such extensions here.
A \emph{bipointed NDA SOS specification} is a bipointed LTS SOS specification together with
a collection of output rules such that for each operator $f$ and for each
$n$-tuple $\bar{o}=(o_1, \ldots, o_n) \in 2^n$, there is at most one output rule triggered by $f$ and
$\bar{o}$. 
Any bipointed NDA SOS specification is easily seen to give rise to a bipointed NDA specification (\ref{eq:nondet}). By
Proposition~\ref{prop:hom} we now have
\begin{corollary}\label{cor-nda-closure}
The operations defined by a bipointed NDA (SOS) specification on the final coalgebra of the $\set$ functor $FX = 2 \times \powf(X)^A$,
where $A$ is a finite set, restrict to the rational fixpoint of $F$.
\end{corollary}
Besides inducing an algebra structure $\alpha:\Sigma(\nu F) \to \nu F$
that restricts to $\rho F$, a bipointed specification as
in~\refeq{eq:nondet} also induces an algebra on formal languages,
i.\,e., $\tilde\alpha: \Sigma (\nu G) \to \nu G$ for $GX = 2 \times
X^A$ on $\set$. To see this recall from Examples~\ref{ex:coalg}(4)
and~\ref{ex:rat}(4) the descriptions of $\nu F$ (and $\rho F$) as 
(rational) strongly extensional trees. Now consider the following map $s: \nu G \to
\nu F$: it takes a formal language $L$ and first interprets its
characteristic map $A^* \to 2$ as a complete ordered
$|A|$-ary tree $t_L$ with nodes labelled in $2$; the strongly
extensional tree $s(L)$ is then obtained by forgetting the order on
the children of every node of $t_L$ and labelling the outgoing edges
of every node with the corresponding letter from $A$. So $s(L)$ has
the same shape as $t_L$, and every node of $s(L)$ has for every $a \in
A$ precisely one $a$-labelled edge to a successor node. 
Secondly, let $q: \nu F \to \nu G$ be the map that assigns to every
strongly extensional tree $t$ in $\nu F$ its corresponding formal
language of all words given by paths from the root of $t$ to a node
labelled by $1$. 
\takeout{
given a formal language $L: \Sigma^* \to 2$, consider the
countable $F$-coalgebra $\langle L, t\rangle: \Sigma^* \to 2 \times
(\powf\Sigma^*)^A$ with $t(w)(a) = \{wa\}$ and form the quotient $Q$
modulo the largest bisimulation. Then $Q$ is a subcoalgebra of $\nu F$ and we define
$s(L)$ to be the equivalence class in $Q \subseteq \nu F$ of the
empty word in the $F$-coalgebra $(\Sigma^*, \langle L,
t\rangle)$. Secondly, let $q: \nu F \to \nu G$ be the map that assigns 
to every equivalence class $[s]$ of a state in the coproduct of all
countable non-deterministic $F$-coalgebras the language accepted by
the state $s$; this is well-defined since bisimilar states always accept the
same language. 
} %
Clearly, we have $q \o s = id_{\nu G}$. Now define
\[
\tilde \alpha = (\xymatrix@1{
\Sigma(\nu G) \ar[r]^-{\Sigma s} & \Sigma(\nu F) \ar[r]^-{\alpha} &
\nu F \ar[r]^-{q} & \nu G
}).
\]
\takeout{
Observe that for a regular language the above quotient $Q$ is
finite, and $q$ maps equivalence classes of states of finite automata to regular
languages. Thus, $s$ and $q$ restrict to the corresponding rational
fixpoints and we have
} %
Observe that $s$ maps a regular language to a regular tree in $\nu
F$, and $q$ maps a regular tree in $\nu F$ to a regular
language. Thus, $s$ and $q$ restrict to the corresponding rational
fixpoints and we have
%\smnote{Cor.~5.4. is not immediate. Proof needed!}
\begin{corollary}\label{cor-reg-closure}
  The set of regular languages over a finite alphabet $A$ is closed under any operation 
  defined in a bipointed NDA (SOS) specification.
\end{corollary}

More precisely, the above algebra structure $\tilde\alpha: \Sigma(\nu G) \to
\nu G$ restricts to an algebra structure $\tilde\beta: \Sigma(\rho G) \to
\rho G$ on the rational fixpoint (i.\,e., on regular languages) with $\tilde\beta = q' \o \beta \o
s'$, where $q'$, $\beta$ and $s'$ are the restrictions of $q$,
$\alpha$ and $s$, respectively, to the rational fixpoints $\rho F$ and
$\rho G$. 

\takeout{ % From Stefan: I deleted this. This will be done working
          % with the category of join-semilattices.
\begin{remark}\label{rem-alg}
  The above reasoning defining the algebra structure $\tilde\alpha$ on
  $\nu G$ that restricts to the rational fixpoint $\rho G$ can be
  performed much more generally in the setting of the 
  generalized power-set construction of~\cite{bbrs_fsttcs} (see
  also~\cite{BMS12}). This will be presented in future work.
  % For lack of space we cannot give this general treatment here. 
  %\jrnote{We put this as future work now, so should we remove the last sentence
  %of the remark?}\smnote{Changed.}
\end{remark}} % end takeout

Given two words $w$ and $v$, the \emph{shuffle} of $w$ and $v$, denoted $w \bowtie v$, is the set of words obtained
by arbitrary interleavings of $w$ and $v$~\cite{Shallit08}. For example, $ab \bowtie c = \{abc, acb, cab\}$. The shuffle
of two languages $L_1$ and $L_2$ is the pointwise extension: $L_1 \bowtie L_2 = \bigcup_{w \in L_1, v \in L_2} w \bowtie v$.
The shuffle operator can be defined in terms of a bipointed NDA SOS specification as follows:
$$
\frac{s \stackrel{a}{\rightarrow} s'}
{s \bowtie t \stackrel{a}{\rightarrow} s' \bowtie t}
\qquad
\frac{t \stackrel{a}{\rightarrow} t'}
{s \bowtie t \stackrel{a}{\rightarrow} s \bowtie t'}
\qquad
\frac{s \downarrow \quad t \downarrow}
{(s \bowtie t) \downarrow}
$$
By Corollary~\ref{cor-nda-closure}, this operation restricts to the rational fixpoint
of non-deterministic automata, and by Corollary~\ref{cor-reg-closure} we obtain the
fact that regular languages are closed under shuffle.

The \emph{perfect shuffle} of two words $w$ and $v$ of the same length is defined as the alternation between 
the two words, reminiscent of the \emph{zip} operation on streams
discussed above~\cite{Shallit08}. The operation assigning to two
formal languages the language of all perfect shuffles of their words can also
easily be defined as a bipointed specification; in fact it can be
defined using a bipointed specification w.r.t.~the type functor $G$ of deterministic automata.

\takeout{
\paragraph{Probabilistic transition systems}\smnote{Maybe ``future work''?}
(We may include probabilistic transition systems.. but then, it may become a bit much.
Different notions of parallel composition: see p. 148 of Bartels' thesis.)
}

\vspace*{-10pt}
\paragraph{Weighted transition systems.}

Recall from Example~\ref{ex:coalg}(5) that weighted transition systems are coalgebras for the functor 
$FX = (\mathcal{F}_{\M}X)^A$ on $\set$; here, we assume $A$ to be finite.
%For labelled transition systems, the format instantiates to
In this case a bipointed specification is a natural transformation with components
\begin{equation}\label{eq:wts-dl}
\lambda_X: \Sigma((\mathcal{F}_{\M}X)^A \times X) \Rightarrow
(\mathcal{F}_{\M}(\Sigma X + X))^A.
\end{equation}
We call these natural transformations \emph{bipointed WTS specifications}.
A general GSOS format for weighted transition systems is given in~\cite{Klin09}. We restrict it to bipointed specifications
as follows. A \emph{bipointed WTS SOS rule} for an operator $f$ in $\Sigma$ of arity $n$ is defined as 
\begin{equation}
\frac{\{x_{i_j} \stackrel{a_j, u_j}{\longrightarrow} y_j\}_{j = 1..m} \qquad \{x_i \stackrel{a}{\Rightarrow} w_{a,i}\}_{a \in D_i, i = 1..n}}
{f(x_1, \ldots, x_n) \xrightarrow{c,~\beta(u_1, \ldots, u_k)} t}
\end{equation}
where $m$ is the number of weighted transitions in the premise. The variables $x_1, \ldots, x_n, y_1, \ldots, y_m$
are again pairwise distinct; let $V$ be the set consisting of these variables. Then $t$ is either a variable in $V$ or a flat term
$g(z_1, \ldots, z_p)$, where $g$ is an $p$-ary operation symbol in
$\Sigma$ and $z_1, \ldots, z_p \in V$.
%\smnote{I corrected to \emph{flat} terms.}%
%a term in $\Sigma$ over variables of $V$. 
Further $D_i \subseteq A$ is a subset of labels for which the \emph{total weight} of the
outgoing transitions from $x_i$ is specified by $w_{a,i}$. 
Finally $a_1, \ldots, a_m, c \in A$ are labels, $u_1, \ldots, u_m$ are weight variables, and 
$\beta: \M^n \rightarrow \M$ is a multi-additive function. 
% The above rule
% is \emph{triggered} by an $n$-tuple $(E_1, \ldots, E_n) \subseteq A^n$ if for each $E_i: a_j \in E_{i_j}$ for all $j = 1..m$, 
% and $b_k \not \in E_{i_k}$ for all $k=1..l$. 
A bipointed WTS SOS specification then is a collection of rules of the
above type such that only finitely many rules share the same operator
$f$ in the source, the same label $c$ in the conclusion, and the same
partial function from $\{1, \ldots , n \} \times A$ to $\M$ arising
from their sets of total weight premises~\cite{Klin09}.  Each
bipointed WTS SOS specification induces a distributive law as in
(\ref{eq:wts-dl}) (but the converse does not hold,
see~\cite{Klin09}). So by Proposition~\ref{prop:hom} we have
\begin{corollary}\label{cor-lts-closure}
The operations defined by a bipointed WTS (SOS) specification on the final coalgebra of the $\set$ functor $FX = (\mathcal{F}_{\M}X)^A$
where $A$ is a finite set, restrict to the
rational fixpoint of $F$, i.e., the coalgebra of all \emph{finite}
weighted transition systems modulo weighted bisimilarity. % observational equivalence.
\end{corollary}
All of the examples of operations on WTS's from~\cite{Klin09} are bipointed specifications, from which it follows
that the rational fixpoint is closed under those operations. We recall here the priority operator. To this
end we consider the weights to be in $\real^{+\infty}$, which is the set consisting of all positive reals augmented with infinity 
(denoted $\infty$). By taking \emph{minimum} as the sum operation, this forms a monoid with $\infty$ as the unit. The unary 
operation $\partial_{ab}$ is defined by the rules
$$
\frac{x \stackrel{a}{\Rightarrow} w \quad x \stackrel{b}{\Rightarrow} v \quad x \stackrel{a,u}{\rightarrow} x'}
{\partial_{ab}(x) \stackrel{a,u}{\longrightarrow} \partial_{ab}(x')}
\qquad
\frac{x \stackrel{a}{\Rightarrow} v \quad x \stackrel{b}{\Rightarrow} w \quad x \stackrel{b,u}{\rightarrow} x'}
{\partial_{ab}(x) \stackrel{b,u}{\longrightarrow} \partial_{ab}(x')}
$$
for all $w \leq v \in \real^{+\infty}$. The operator $\partial_{ab}$ preserves only the $a$-transitions 
if the minimum weight of all $a$-transitions is less than or equal to the minimum of all outgoing $b$-transitions,
and vice versa. 

\takeout{
\paragraph{Systems on other categories}
(A discussion on, for example, determinstic automata on join-semilattices and
on vector spaces. These have very nice rational fixpoints, but the associated
natural transformations are a bit more nasty, and a little less expressive. 
But at least it might be interesting to say something about this, maybe
as future work. For example, would there be concrete formats which guarantee
that the arrows of the induced natural transformation are homomorphisms? 
This is also a bit outside the scope of this paper i would say. )
}

\section{Conclusions and future work}\label{sec-conc}

In this paper we have presented a general categorical framework for the specification
of algebraic operations on regular behaviour based on distributive laws. The theory we have
presented works not only in $\set$ but also in many other categories including
vector spaces and other algebraic categories. In this paper we have instantiated
the general theory to several concrete specification formats in $\set$. It 
remains an interesting challenge to study concrete formats for distributive laws
on other categories, not only for our bipointed specifications but also for
distributive laws corresponding to GSOS. For example, working out a
format for the functor $FX = 2 \times X^A$ on the category of
join-semilattices will give a more direct way to define operations
like the shuffle product of formal languages which cannot be captured
by a bipointed specification for $F$ on $\set$. %
%In future work we will combine 
%distributive laws with the generalized powerset construction~\cite{bbrs_fsttcs}, as
%already mentioned in Remark~\ref{rem-alg}.
Finally, it is interesting to study extensions of the format introduced 
in this paper. We already mentioned the coGSOS format, and we will
investigate this more thoroughly in the future. One would also hope
for formats covering all the standard operations on formal languages
such as the Kleene star which, presently, does not arise as an
application of our theory. Since checking if a specification gives rise to
operations under which regular behaviour is closed
is in general undecidable, a complete format cannot
exist~\cite{AFV}.

%\vspace*{-10pt}
%\enlargethispage{10pt}
\bibliographystyle{eptcs}
\bibliography{ops-rfp}

\end{document}